\documentclass[12pt,draftcls,onecolumn]{IEEEtran}
\usepackage{amsfonts,epsfig,amsmath,latexsym,amssymb,amscd,multirow,graphicx,geometry,lscape,amsmath,amssymb,bm,pifont,graphicx,amssymb,amsmath}
\usepackage{amsthm}
\usepackage[]{cite}
\usepackage{cancel,algorithm,algorithmic,verbatim,array,multicol,palatino,amsfonts}
\usepackage[usenames,dvipsnames]{xcolor}
\usepackage{graphicx}
\usepackage{lineno}
\usepackage{subfig}
\usepackage[normalem]{ulem}
\usepackage{dsfont}
\usepackage{framed}

\newtheorem{theorem}{\noindent \textbf{Theorem}}

\newtheoremstyle{remarkstyle}
  {\topsep}   
  {\topsep}   
  {\slshape}  
  {0pt}       
  {\bfseries} 
  {.}         
  {5pt plus 1pt minus 1pt} 
  {}          

\theoremstyle{remarkstyle}

\newtheorem{definition}{\noindent \textbf{Definition}}

\modulolinenumbers[1]

\begin{document}
\title{Low-High SNR Transition in Multiuser MIMO}
\author{Malcolm Egan$^{1}$

\begin{center}
{\footnotesize {\
\textit{
$^{1}$ Agent Technology Center, Faculty of Electrical Engineering, Czech Technical University in Prague, Czech Republic.\\[0pt]
} } }
\end{center}
}

\maketitle

\vspace{-2cm}

\begin{abstract}
%

Multiuser MIMO (MU-MIMO) plays a key role in the widely adopted 3GPP LTE standard for wireless cellular networks. While exact and asymptotic sum-rate results are well known, the problem of obtaining intuitive analytical results for medium signal-to-noise ratios (SNRs) is still not solved. In this paper, we propose the bend point, which quantifies the transition between low and high SNR; i.e., the beginning of the high SNR region. We derive the bend point for MU-MIMO with zero-forcing precoding and show that it is intimately related to the intercept of the high SNR asymptote with the zero sum-rate line. Using this result, we obtain a new approximation of the sum-rate at the bend point---providing a useful rule of thumb for the effect of increasing the number of antennas at medium SNR.

\end{abstract}
\newpage
\section{Introduction} \label{Introduction}

With the widespread adoption of the 3GPP LTE and development of the 3GPP LTE-A standards in wireless cellular networks, network MIMO--particularly, coordinated multipoint \cite{Lee2012}---is a key technique for mitigating interference between base stations (BSs). The basis of network MIMO is multiuser MIMO (MU-MIMO) \cite{Ghosh2012}, where a BS transmits data to multiple users simultaneously by exploiting multiple antennas.

A popular and practical signal processing technique for MU-MIMO is zero-forcing (ZF) precoding \cite{Caire2003}, where interference between users is cancelled. Although ZF precoding is suboptimal, it has the desirable feature that it requires only matrix inversion to implement.

Asymptotic---at high signal-to-noise ratios (SNR)---and exact analysis of ZF precoding is well-established \cite{Peel2005}, as well as antenna switching schemes \cite{Gesbert2007}. However, the problem of obtaining intuitive analytical results at medium SNR has not been solved. This problem is important as medium SNR is typically where practical systems operate and intuitive results are useful to provide rules of thumb for design.

In this paper, we propose the bend point as a new intuitive way of quantifying the transition between low and high SNR. In particular, the bend point is defined as the SNR where the second derivative of the sum-rate is maximized. In fact, the bend point can be interpreted as the beginning of the high SNR region.

We obtain a simple and accurate closed-form approximation for the bend point of the sum-rate for MU-MIMO with ZF precoding. We then show that in Rayleigh fading, the bend point increases as the number of antennas increases. Moreover, the sum-rate is approximately equal to the number of antennas times a constant, at the bend point. As such, the bend point provides a useful rule of thumb for the effect of increasing the number of antennas at medium SNR.

\section{System Model}

Consider the MU-MIMO network consisting of a BS with $N$ antennas and $N$ single antenna users. The received signal is denoted by $\mathbf{y} = [y_1,\ldots,y_N]^T$, given by
\begin{align}
\mathbf{y} = \mathbf{H}\mathbf{V}\mathbf{s} + \mathbf{n},
\end{align}
where $\mathbf{V} = \sqrt{P}[\mathbf{v}_1,\ldots,\mathbf{v}_N]$ is the precoder, $P$ is the transmit power, $\mathbf{H} = [\mathbf{h}_1,\ldots,\mathbf{h}_N]$ is the $N\times N$ channel matrix for which $\mathbf{h}_i^\dag$ corresponds to the channel vector of the $i$-th user, $\mathbf{s}$ denotes the Gaussian data symbols, $\mathbf{n}$ is the noise vector with each element distributed as $n_i \sim \mathcal{CN}(0,\sigma^2)$, and $\sigma^2$ is the noise power at each user. We also define the SNR as $\rho = 10\log_{10}\left(\frac{P}{\sigma^2}\right)$, which is closely related to the per-user SNR in decibels (dB).

We assume that there is perfect channel state information from each user to the BS. Although in practice there is limited feedback, the quantization error is small for sufficiently large quantization codebooks \cite{Egan2014}. We use the ZF precoder to suppress inter-user interference, which is given by
\begin{align}
\mathbf{V} = \frac{\mathbf{H}^{-1}}{\|\mathbf{H}^{-1}\|_F},
\end{align}
where $\|\cdot\|_F$ is the Frobenius matrix norm.

The sum-rate in nats of the network is then given by
\begin{align}
R = N\log\left(1 + \frac{10^{\rho/10}}{\eta}\right),
\end{align}
where $\eta = \|\mathbf{H}^{-1}\|_F^2$. It is important to note that $\rho$ is in dB, which is assumed throughout the remainder of the paper.

\section{The Bend Point}

In this section, we introduce the bend point to quantify the transition of the sum-rate from low to high SNR. We also show the intimate relationship between the bend point and the intercept of the high SNR asymptote with $R = 0$.

First, the bend point is defined as follows.
\begin{definition}[Bend Point]
The bend point, $\rho_{bend}$, of the sum-rate $R$ is the SNR $\rho$ such that the second derivative $R''(\rho)$ is maximized.
\end{definition}
Intuitively, the bend point can be interpreted as the transition point between high and low SNR, or as the beginning of the high SNR region. This is due to the fact that the bend point is where the variation of the rate of change of the slope of $R$ is maximized. At the bend point, the sum-rate transitions from the low SNR asymptote (i.e., $R = 0$) to the linear high SNR asymptote.

Next, we derive an exact analytical expression for the bend point when the channel is deterministic (i.e., $\mathbf{H}$ is fixed).
\begin{theorem}\label{thrm:det_bend}
The unique sum-rate bend point $\rho_{bend}$ in MU-MIMO with ZF precoding is given by
\begin{align}
\rho_{bend} = \rho_{int} = \frac{10\log \eta}{\log 10},
\end{align}
where $\rho_{int}$ is the intercept of the high SNR asymptote with $R = 0$.
\end{theorem}
\begin{proof}
First, the asymptotic sum-rate is given by
\begin{align}
R^{\infty} = N\log\left(\frac{10^{\rho_{int}/10}}{{\eta}}\right),
\end{align}
and the $R = 0$ intercept, $\rho_{int}$, satisfies
\begin{align}
N\log\left(\frac{10^{\rho_{int}/10}}{{\eta}}\right) = 0~~~
\Rightarrow \rho_{int} = \frac{10\log \eta}{\log 10}.
\end{align}
Next, the third derivative of $R$ is given by
\begin{align}\label{eq:second_deriv}
R'''(\rho) = \frac{N}{\eta^2}\left(\frac{\log 10}{10}\right)^3\left(\frac{\eta10^{\rho/10} - 10^{2\rho/10}}{(1 + \frac{10^{\rho/10}}{\eta})^3}\right).
\end{align}
Substituting $\rho_{int}$ into (\ref{eq:second_deriv}) yields $R'''(\rho_{int}) = 0$, which implies that $\rho_{int}$ is a stationary point. It is straightforward to check that the fourth derivative satisfies $R''''(\rho_{int}) < 0$, which implies $\rho_{int}$ is a local maximum of $R''(\rho)$. To prove that $\rho_{bend} = \rho_{int}$ is the unique maximum, we note that the second derivative $R''(\rho)$ is symmetric around $\rho_{bend}$ and monotonically decreasing for $\rho > \rho_{bend}$.
\end{proof}

\begin{figure}[h]
\centering{\includegraphics[width=90mm]{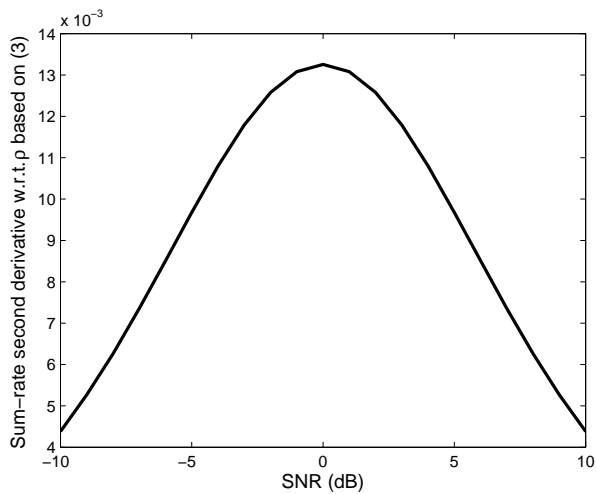}}
\caption{Plot of the sum-rate second derivative $R''(\rho)$ versus the SNR $\rho$, for $\eta = 1$.}\label{fig:2nd_deriv}
\end{figure}

The behavior of the second derivative $R''(\rho)$ is illustrated in Fig.~\ref{fig:2nd_deriv}. Observe that $\rho_{bend} = \rho_{int}$ is the unique maximum (which in this case corresponds to $\rho_{bend} = \rho_{int} = 0$), as expected.

We now consider the ergodic sum-rate with ZF precoding in Rayleigh fading. In this case, $\mathbf{H}$ is a random matrix with elements $h_{ii} \sim \mathcal{CN}(0,1)$, and $\eta$ is a random variable. In \cite{Peel2005}, it was shown that for a sufficiently large number of users, the ergodic sum-rate is well approximated by
\begin{align}\label{eq:ergodic}
R_{E} &= \int_0^\infty \log\left(1 + \frac{10^{\rho/10}\eta}{N}\right)e^{-\eta}d\eta
\notag\\
&= Ne^{N/10^{\rho/10}}E_1\left(\frac{N}{10^{\rho/10}}\right),
\end{align}
where $E_1(\cdot)$ is the exponential integral, given by
\begin{align}
E_1(x) = \int_x^\infty \frac{e^{-t}}{t}dt.
\end{align}
We now seek to find the bend point $\rho_{E,bend}$. Motivated by Theorem~\ref{thrm:det_bend}, we consider the high SNR asymptote for fixed $N$, which is given by
\begin{align}
R^{\infty}_{E} = N\left(\frac{\rho}{10}\log10 - \gamma - \log N\right),
\end{align}
where $\gamma$ is the Euler-Mascheroni constant. This holds since \cite{Abramowitz1985},
\begin{align}
E_1(x) \cong -\gamma - \log x,~~~\textrm{as}~x \rightarrow 0.
\end{align}
As such, the intercept $\rho_{E,int}$ of the high SNR asymptote with the line $R_E = 0$ is given by
\begin{align}\label{eq:erg_int}
\rho_{E,int} = .\frac{10(\gamma + \log N)}{\log 10}.
\end{align}
Importantly, the ergodic sum-rate at $\rho_{E,int}$ is given by
\begin{align}\label{eq:int_rate}
R_E(\rho_{E,int}) &= N\exp\left(10^{-\frac{\gamma}{\log 10}}\right)E_1\left(10^{-\frac{\gamma}{\log 10}}\right)\notag\\
&\approx 0.86N,
\end{align}
which is the number of antennas times a constant. This means that when the number of antennas is increased by one, the difference between the sum-rate---evaluated at the respective bend points--is approximately $0.86~$nats. This is illustrated in Fig. \ref{fig:rate_dif}, where the dots correspond to the intercepts $\rho_{E,int}$ for $N = 1$ and $N = 2$. Importantly, there is a gap of approximately $0.86~$nats between sum-rate at each intercept point, in agreement with (\ref{eq:int_rate}).

\begin{figure}[h]
\centering{\includegraphics[width=90mm]{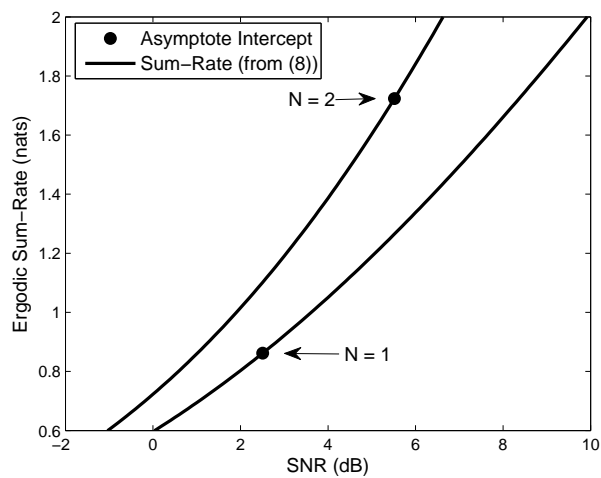}}
\caption{Plot of the ergodic sum-rate in Rayleigh Fading versus the SNR, for $N = 1$ and $N = 2$ antennas.}\label{fig:rate_dif}
\end{figure}

\section{Numerical Results and Discussion}

In this section, we show numerically that the bend point is well approximated by the asymptote intercept $\rho_{E,int}$, for the Rayleigh fading channel. This means that the sum-rate difference after increasing the number of antennas by one is approximately $0.86~$nats, between the resulting bend points; obtained via the asymptote intercept result in (\ref{eq:int_rate}).

\begin{figure}[h]
\centering{\includegraphics[width=90mm]{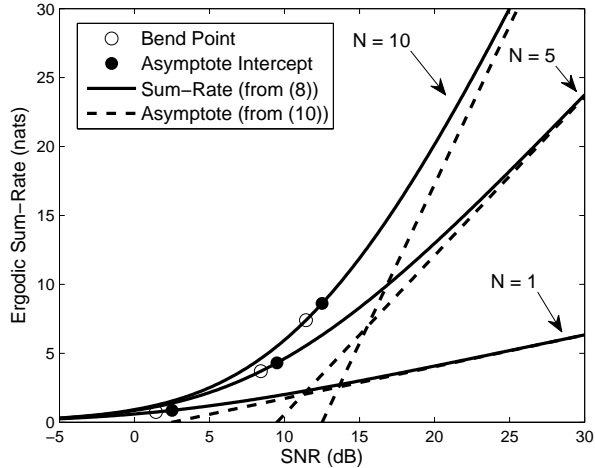}}
\caption{Plot of the ergodic sum-rate in Rayleigh fading (from (\ref{eq:ergodic})) versus the SNR $\rho$.}\label{fig:ergodic_rate}
\end{figure}

Fig.~\ref{fig:ergodic_rate} plots the sum-rate in nats versus the SNR, for the Rayleigh fading channel. The key observation is that the bend point $\rho_{E,bend}$ and the asymptote intercept $\rho_{E,int}$ are approximately equal. Moreover, the bend point increases as the number of antennas increases, which was expected from (\ref{eq:erg_int}).

As the bend point is well approximated by the asymptote intercept, it means that the sum-rate at the bend points for different numbers of antennas is approximately given by (\ref{eq:int_rate}). In particular, increasing the number of antennas by one leads to an increase in the sum-rate of $0.86~$nats, between the resulting bend points. This provides a useful rule of thumb for the effect of increasing the number of antennas at medium SNR.

\section{Conclusion}

We have introduced the bend point to quantify the transition between low and high SNR in MU-MIMO with ZF precoding. We have shown that there is an intimate connection between the bend point and the intercept of the high SNR asymptote with $R = 0$. Using this result, we then provided a rule of thumb for the effect of increasing the number of antennas at medium SNR.

\bibliographystyle{ieeetr}
\bibliography{bendpoint}

\end{document}